\documentclass[12pt,a4paper,onecolumn]{IEEEtran}

\usepackage{amsopn,amsmath,amssymb,amsfonts,bm,amsthm}
\usepackage{graphicx,subfigure}
\usepackage{caption,setspace} 
\usepackage{enumerate}
\usepackage{cite}
\usepackage{flushend}
\usepackage{etex}
\usepackage{algorithm,algorithmic}
\usepackage{tikz,pgfplots}
\usepackage[margin=16mm,top=20mm,bottom=16mm]{geometry}
 \usepackage{setspace}
\let\Algorithm\algorithm
\renewcommand\algorithm[1][]{\Algorithm[#1]\setstretch{2}}

\newcommand{\figref}[1]{Fig.\,\ref{#1}}

\def\a{\alpha}
\def\mA{\boldsymbol A}

\def\bu{\boldsymbol u}
\def\cc{\boldsymbol c}
\def\Dmin{D_{\min}}

\def\ud{\underline{D}}

\def\E{\mathcal E}
\def\e{\boldsymbol e}

\def\F{\mathbb F}

\def\G{\mathcal G}
\def\mH{\mathcal H}
\def\M{\mathcal M}

\def\P{\mathcal{P}}
\def\p{\mathbf{p}}
\def\R{\mathcal{R}}
\def\S{\mathcal S}

\def\Umin{U_{\min}}
\def\W{\mathcal W}
\def\w{\mathcal W}

\def\X{\boldsymbol X}

\def\V{\mathcal V}

\newtheorem{Lemma}{\textbf{Lemma}}
\newtheorem{Theorem}{\textbf{Theorem}}
\newtheorem{Corollary}{Corollary}
\newtheorem{Definition}{\textbf{Definition}}
\newtheorem{Remark}{\textbf{Remark}}

\newtheorem{Example}{\textbf{Example}}
{\proof}{\proofend}

\newtheorem{property}{Property}

\title{\huge{On the Packet Decoding Delay of Linear Network Coded Wireless Broadcast}}
\author{Mingchao Yu, Alex Sprintson$^\dag$, and Parastoo Sadeghi* \\ \small{*Research School of Engineering, Australian National University, Canberra, Australia\\$^\dag$Department of Electrical and Computer Engineering, Texas A\&M University, Texas, USA\\ \texttt{Emails: \{ming.yu,parastoo.sadeghi\}@anu.edu.au,~spalex@tamu.edu}}}
\begin{document}
\maketitle
\thispagestyle{empty}
\vspace{-2em}

\begin{abstract}
	We apply linear network coding (LNC) to broadcast a block of data packets from one sender to a set of receivers via lossy wireless channels, assuming each receiver already possesses a subset of these packets and wants the rest. We aim to characterize the average packet decoding delay (APDD), which reflects how soon each individual data packet can be decoded by each receiver on average, and to minimize it while achieving optimal throughput. To this end, we first derive closed-form lower bounds on the expected APDD of all LNC techniques under random packet erasures. We then prove that these bounds are NP-hard to achieve and, thus, that APDD minimization is an NP-hard problem. We then study the performance of some existing LNC techniques, including random linear network coding (RLNC) and instantly decodable network coding (IDNC). We proved that all throughput-optimal LNC techniques can approximate the minimum expected APDD with a ratio between 4/3 and 2. In particular, the ratio of RLNC is exactly 2. We then prove that all IDNC techniques are only heuristics in terms of throughput optimization and {cannot guarantee an APDD approximation ratio for at least a subset of the receivers}. Finally, we propose hyper-graphic linear network coding (HLNC), a novel throughput-optimal and APDD-approximating LNC technique based on a hypergraph model of receivers' packet reception state. We implement it under different availability of receiver feedback, and numerically compare its performance with RLNC and a heuristic general IDNC technique. The results show that the APDD performance of HLNC is better under all tested system settings, even if receiver feedback is only collected intermittently.
	
	\vspace{1em}
	\noindent
	\textbf{\emph{Keywords}}: \emph{wireless broadcast, linear network coding, delay, combinatorial optimization, hypergraph coloring}
\end{abstract}

\section{Introduction}
\subsection{Background}
In this paper, we consider a wireless broadcast problem where a sender wishes to broadcast a block of $K$ data packets to a set of $N$ receivers using linear network coding (LNC) \cite{Yeung_flow,li2003linear,koetter2003algebraic,ho:medard:koetter:karger:effros:2006,nistor:lucani:vinhoza:costa:barros:2011}. Each receiver is assumed to already possess a subset of the data packets and still wants all the remaining data packets. The sender transmits, as LNC coded packets, linear combinations of the data packets, so that the receivers can decode their missing data packets through solving linear equations. We are interested in studying two interrelated performance metrics, namely throughput and average packet decoding delay (APDD) across the receivers: Throughput measures the time it takes to deliver the whole data block to all receivers. APDD measures the average time it takes each receiver to decode each data packets in the block. For example:

\begin{Example}
	Consider two receivers, $ r_1$ and $ r_2$, and a block of three data packets $\{\p_1,\p_2,\p_3\}$. $ r_1$ already has $\p_1$ and $ r_2$ already has $\p_2$. They both want the remaining two data packets. Compare two LNC schemes:
	\begin{itemize}
		\item Scheme-1: send $\p_1+\p_2$, and then $\p_3$;
		\item Scheme-2: send $\p_1+\p_2+\p_3$, and then $\p_1+2\p_2+3\p_3$.
	\end{itemize}
	Both schemes allow the two receivers to fully decode after two transmissions. Thus, they offer the same throughput. However, they offer different APDD: Scheme-1 allows both receivers to decode one data packet after both the first and second transmissions, yielding an APDD of $\frac{1+1+2+2}{4}=1.5$. But Scheme-2 does not allow any packet decodings after the first transmission, yielding a larger APDD of $\frac{2+2+2+2}{4}=2$.
\end{Example}

A lower APDD implies faster and smoother data delivery  to the application layer within each data block, and is particularly important in applications where individual data packets are informative. Thus, we are interested in the the minimization of APDD and its relation to throughput. However, as we will review next, APDD minimization is an open problem in the literature.

\subsection{Some Existing LNC Classes}
A well-known class of LNC techniques, which will be studied in this paper, is the class of throughput-optimal LNC techniques, whose packets are innovative (bring new information) to every receiver who has not yet fully decoded all the packets in the block. Throughput-optimal packets can be generated either randomly (i.e., the classic random LNC (RLNC) technique \cite{ho:medard:koetter:karger:effros:2006,lucani:tdd_field:2009,heide_systematic_RLNC}) or deterministically ({e.g., by solving a hitting set problem \cite{kwan2011generation}, or by solving a matroid graphic representation problem \cite{yu2014deterministic}}). All these throughput-optimal techniques could suffer from large APDD due to block decoding: each receiver is only able to decode all data packets at once after receiving sufficient linearly independent coded packets. In other words, there is generally no early packet decodings, which would  help reduce APDD. {In order to provide early packet decodings without sacrificing throughput, Keller \emph{et al} proposed in \cite{keller2008online} a two-step coding technique that adds extra data packets to an instantly decodable but throughput-suboptimal coded packet for throughput optimality.} But this technique was not analytically studied. Indeed, the APDD of throughput-optimal LNC techniques has not been well studied. Only recently, Yu \emph{et al} proved in  \cite{yu:sprintson:sadeghi:netcod2015} that, when there is no packet erasure, RLNC is an APDD 2-approximation technique. That is, its APDD is at most twice of the minimum possible APDD that any LNC techniques can offer.

Another well-known class of LNC techniques, which will also be studied in this paper, is instantly decodable network coding (IDNC). This class aims at reducing APDD through  early packet decodings \cite{Rozner_Heuristic_clique,sorour:valaee:2010,sadeghi:shams:traskov:2010,sorour2015completion, yu2015SIDNC}. The approach is to send, as coded packets, binary XORs of selected data packets, so that a subset of (or all) receivers can instantly decode one wanted data packet from each coded packet. It is well known that IDNC techniques are generally not throughput optimal \cite{costa:munaretto:widmer:baros:2008,li:idnc_video:2011}.  Moreover, \cite{sorour2015completion} proved that it is intractable to  maximize the throughput subject to IDNC constraints. This suboptimal throughput, in turn, can hurt APDD, for it brings large delays to data packets decoded in the final stage of the broadcast. Although a large body of heuristics have been developed as a remedy \cite{neda:parastoo:o2idnc,neda:parastoo:sameh:balance2013,sorour2015graph,sorour2015completion}, it remains an open problem whether or not IDNC techniques are able to at least approximate the optimal throughput and APDD. Only recently, Yu \emph{et al} proved in \cite{yu:sprintson:sadeghi:netcod2015} that the class of strict IDNC techniques \cite{Rozner_Heuristic_clique,sadeghi:shams:traskov:2010,yu2015SIDNC} are not APDD-approximation techniques even when there are no packet erasures.

There are also LNC techniques that strike a balance between throughput and APDD \cite{yu2013rapprochement,yu:parastoo:neda:2014,yu2016feedback}: They partition the packet block into sub-blocks and broadcast each sub-block separately, so that data packets from earlier sub-blocks can be decoded sooner. However, these techniques are generally heuristic. Thus, they will not be studied in this paper.

We also note that index coding (IC) considers a more general system setting than this paper, which assumes that every receiver has a subset of the $K$ data packets and wants one \cite{sprintson:min:2007,sprintson:algorithm:2008,sprintson:ic_nc_matroid:2010,Yossef:index:2011} or some \cite{fragouli:pliable:ic:2013} of the rest. While throughput optimization has been extensively studied in the IC literature, APDD minimization has not previously been considered in the IC literature to the best of our knowledge. Moreover, most works in the IC literature assume no packet erasures. Therefore, our study may also provide new insights into the corresponding problem in the IC context.

\subsection{Contributions}
In summary, APDD minimization for erasure-prone LNC wireless broadcast systems is a highly nontrivial problem. There have not been any comprehensive studies on the APDD performance of LNC techniques, nor any optimal or approximation techniques, only heuristics. We are thus motivated to close this knowledge gap. Specifically, in this paper we achieve the following main contributions:

\begin{enumerate}
	\item {We reveal the APDD performance limits of LNC techniques by deriving closed-form lower bounds on the expected APDD\footnote{{By ``average'' packet decoding delay, we mean the average of the decoding delay of every data packet at every receiver in a given instance of the problem. By ``expected'' APDD, we mean the statistical average of the APDD of all possible instances of the problem generated under random packet erasures.}} of LNC techniques under random packet erasures;}
	\item We prove that APDD minimization is NP-hard;
	\item We prove that RLNC is an APDD 2-approximate technique. We also prove that all throughput-optimal LNC techniques can approximate the minimum expected APDD with a ratio between 4/3 and 2;
	\item We prove that all IDNC techniques cannot approximate the optimal throughput, {and show that they cannot guarantee an APDD approximation ratio for at least a subset of the receivers}; 
	\item {We propose hypergraphic linear network coding (HLNC), a novel throughput-optimal and APDD-approximation LNC technique built upon a novel hypergraph representation of receivers' knowledge space.} Extensive simulations show that it outperforms RLNC and a heuristic general IDNC technique in terms of APDD under all tested system parameter settings. HLNC is implementation-friendly, for it does not require NP-hard coding decision making nor receiver feedback after every transmission.
\end{enumerate}
The rest of this paper is organized as follows:
\begin{itemize}
	\item Section 2 defines the system model, notations, and terminologies;
	\item Section 3 studies the fundamental features of APDD, including lower bounds on the expected APDD, and the NP-hardness of the APDD minimization problem;
	\item Section 4 studies the performance of RLNC and IDNC techniques;
	\item Section 5 proposes HLNC and discusses its implementations under different availability of receiver feedback;
	\item Section 6 numerically compares the APDD performance of HLNC with RLNC and a general IDNC technique. It also numerically demonstrates the feedback reduction capability of HLNC;
	\item Section 7 concludes the paper.
\end{itemize}

\begin{Remark}
	We note that preliminary results of some of the above contributions have been derived in \cite{yu:sprintson:sadeghi:netcod2015} under the limited setting without packet erasures. As we will see later, their extension to settings with packet erasures is non-trivial, and requires completely different derivation approaches.
\end{Remark}

\section{System Model}
\subsection{System Setting}
We consider a block-based wireless broadcast scenario, in which the sender wishes to deliver a block of $K$ data packets, denoted by $\P=\{\p_k\}_{k=1}^K$, to a set of $N$ receivers, denoted by $\R=\{ r_n\}_{n=1}^N$. All data packets are vectors of the same length, with entries taken from a finite field $\F_q$. Time is slotted, and in each time slot the sender broadcasts a data or coded packet to all receivers. The wireless channel between the sender and each receiver $ r_n$ is independent of each other, and is subject to Bernoulli random packet erasures with a probability of $P_{e,n}\geqslant 0$.

We assume that each receiver has already received a subset of packets in $\P$, and still wants all the rest. Such a packet reception state could be the consequence of previous uncoded transmissions \cite{heide_systematic_RLNC}, and is a common assumption in network coding and index coding literature \cite{sprintson:min:2007,sprintson:algorithm:2008,sprintson:ic_nc_matroid:2010,Yossef:index:2011,fragouli:pliable:ic:2013}. This state can be summarized by a binary $N\times K$ state feedback matrix (SFM) $\mA$. Here $\mA(n,k)=1$ means $ r_n$ has missed $\p_k$ and still wants it, and $\mA(n,k)=0$ means $ r_n$ already has $\p_k$.
The set of data packets wanted by $r_n$ is
denoted by $\W_n$. Its size is denoted by $w_n$. An example of $\mA$ is given in \figref{fig:sfm}. There are 6 data packets and 3 receivers. Receivers $\{ r_1, r_2, r_3\}$ all want 3 data packets. 

The sender then applies an LNC technique to help receivers recover their missing data packets. In each time slot, it broadcasts an LNC packet $\X$, which takes the form of:
\begin{equation}
	\X=\sum_{\p_k\in\M}\a_k\p_k,
\end{equation}
where $\M$ is a subset of $\P$ and is called the coding set of $\X$, and $\{\a_k\}$ are non-zero coding coefficients chosen from $\F_q$, i.e., $\{\a_k\}\subseteq\F_q\setminus\{0\}$. Similarly, $\X$ is called a coded packet of $\M$. Moreover, if $\{\a_k\}$ are chosen uniformly at random from $\F_q\setminus\{0\}$, then $\X$ is called a random-coded packet of $\M$. When $\F_q$ is sufficiently large, any receiver who is missing $k$ data packets in $\M$ can decode them from a set of $k$ random-coded packets of $\M$ w.h.p.


\captionsetup{font={small,stretch=1.5}}
\begin{figure}
	\centering
	\subfigure[SFM $\mA$]{\includegraphics{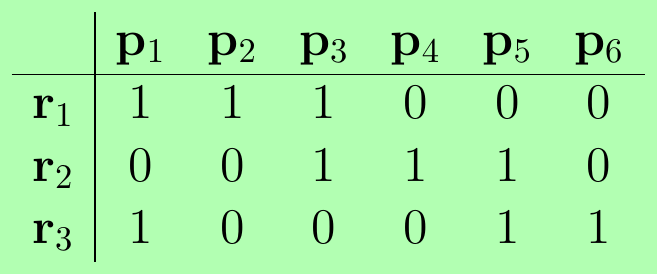}\label{fig:sfm}}\hspace{30pt}
	\subfigure[Hypergraph model $\mH$]{\includegraphics{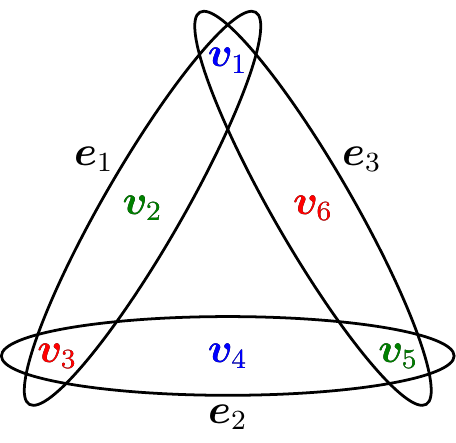}\label{fig:h_color}}
	\caption{An example of state feedback matrix $\mA$ and its hypergraph representation $\mH$. Vertices in $\mH$ are strong-colored, such that all vertices in the same hyperedge have different colors.}
	\label{fig:sfm_h}
\end{figure}

\def\mysum{\mathsf{sum}}

\subsection{Definitions and Terminologies}

Average packet decoding delay (APDD) measures how fast each data packet is decoded on average. Given an instance of the broadcast, the APDD of receiver $ r_n$ is the average time it takes for $ r_n$ to decode a data packet, and is denoted by $D_n$:
\begin{equation}\label{eq:d_def:rec}
	D_n=\frac{1}{w_n}\sum_{\forall k: \p_k\in\W_n}u_{n,k},
\end{equation}
where $u_{n,k}$ is the index of the packet transmission after which $ r_n$ decodes $\p_k$. The APDD across all receivers is similarly defined as
\begin{equation}\label{eq:d_def}
	D=\frac{1}{\mysum(\mA)}\sum_{\forall k,n: \mA(n,k)=1}u_{n,k},
\end{equation}
where $\mysum(\mA)$ is the sum of the entries of $\mA$, and is equal to the number of ones in $\mA$. The expected APDD $E[D_n]$ of a receiver (resp. the expected APDD $E[D]$ across all receivers) under random packet erasures can then be derived by averaging $D_n$ (resp. $D$) over all possible packet erasure patterns.

We further denote by $U_n$ the number of transmissions after which receiver $r_n$ decodes all its wanted data packets. Then, $U\triangleq\max_{n\in[1,N]}U_n$ is the block completion time (BCT). We use BCT to measure system throughput, as throughput is inversely proportional to BCT. The expected $U_n$ and $U$ under random packet erasures are denoted by $E[U_n]$ and $E[U]$, and can be derived by averaging $U_n$ and $U$ over all possible packet erasure patterns.

%

We define an APDD $\beta$-approximation LNC technique as follows:
\begin{Definition}
	An LNC technique is an APDD $\beta$-approximation technique if for any given SFM $\mA$ and packet erasure probabilities $\{P_{e,n}\}_{n=1}^N$, the expected APDD $E[D]$ of using this technique is at most $\beta$ times of the minimum expected APDD $E[D_{\min}]$ that any LNC techniques can offer.
\end{Definition}
Similarly, we define an throughput-approximation LNC techniques as follows:
\begin{Definition}
	An LNC technique is a throughput $\beta$-approximation technique if for any given SFM $\mA$ and packet erasure probabilities $\{P_{e,n}\}_{n=1}^N$, the expected BCT $E[U]$ of using this technique is at most $\beta$ times of the minimum expected BCT $E[U_{\min}]$ that any LNC techniques can offer.
\end{Definition}

\section{{Fundamental APDD Performnace Limits}}\label{sec:approx}
In this section, we will derive lower bounds on the expected APDD of LNC techniques. We will then prove that such lower bounds are NP-hard to achieve, which will indicate that APDD minimization is an NP-hard problem.

Our approach is to investigate the performance and existence of a ``perfect LNC technique'': 
\begin{Definition}
	A LNC technique is perfect if it allows every receiver who is still missing packet(s) to decode one new wanted data packet whenever this receiver successfully receives a coded packet generated using this technique.
\end{Definition}

For example, for the $\mA$ in \figref{fig:sfm_h}(a), a perfect LNC technique could send 3 coded packets when there are no packet erasures: $\X_1=\p_1\oplus\p_4$, $\X_2=\p_2\oplus\p_5$, and $\X_3=\p_3\oplus\p_6$, where $\oplus$ is the binary XOR operation. $\{\X_u\}_{u=1}^3$ allow every receiver to decode one wanted data packet in every transmission. For example, $ r_1$ can decode $\p_1$ from $\X_1$ by performing $\X_1\oplus\p_4$, as it already has $\p_4$.

\subsection{Lower Bound on The Expected APDD}
It is intuitive that the APDD of the perfect LNC technique is a lower bound on the APDD of LNC techniques. We thus denote by $\ud_n$ (resp. $\ud$) the APDD of receiver $r_n$ (resp. across all receiver) using the perfect LNC technique. Their expectations under random packet erasures, namely, $E[\ud_n]$ and $E[\ud]$, are thus the lower bounds on the corresponding expectations across all LNC techniques.

The value of $E[\ud_n]$ is stated in the following theorem:
\begin{Theorem}\label{theo:ud_n}
	When coded transmissions are subject to
	random packet erasures with a probability of $P_{e,n}$, the expected APDD of a receiver $ r_n$ who wants $w_n$ data packets is lower bounded as:
	\begin{equation}\label{eq:ud_n}
		E[\ud_n|\text{when~}r_n~\text{wants}~w_n~\text{data~packets}]=\frac{w_n+1}{2(1-P_{e,n})}
	\end{equation}
\end{Theorem}
The proof is given in Appendix \ref{ap:ud_n}. In the rest of the paper, we will simplify the notation for such conditional expectations to a form of $E[A|(b,c,\cdots)]$, which stands for the expectation of random variable $A$ when the values of some random variables  are given as $b,c,\cdots$. For example, the expectation in the above equation can be simplified to $E[\ud_n|w_n]$.

Then, by averaging the APDD of all receivers with their $w_n$ as weights, we obtain a lower bound on the expected APDD $D$ of any $\mA$ under any $\{P{e,n}\}$:
\begin{Corollary}
	The expected APDD of a wireless broadcast instance with given $\mA$ and $\{P_{e,n}\}$ is lower bounded as:
	\begin{equation}
		E[\ud|\{w_n\}_{n=1}^N]=\frac{\sum_{n=1}^N\frac{w_n^2+w_n}{2(1-P_{e,n})}}{\sum_{n=1}^N{w_n}}\label{eq:ud_a},
	\end{equation}
	where $\{w_n\}$ is obtained from $\mA$.
\end{Corollary}

Then, by assuming that $\mA$ is the result of one round of uncoded transmission of the $K$ data packets and that all receivers experience the same packet erasure probability $P_e$, we obtain a closed-form relation between the APDD performance limit of LNC techniques and system parameters:
\begin{Theorem}\label{theo:ud_overall}
	
	Given system parameters $K$, $N$, and $P_e$, the overall APDD performance of every NC technique is lower bounded by $E[\ud]$, where
	\begin{align}
		E[\ud]
		&=\frac{1}{2(1-P_e)}\left(1+E\left[\frac{\sum_{n=1}^N w_n^2}{\sum_{n=1}^N w_n}\right]\right),~~~~\{w_n\}_{n=1}^N\sim B(K,P_e)\label{eq:ud_overall}\\
		&\approx\frac{KP_e-P_e+2}{2-2P_e},~~\mathrm{when~}N~\mathrm{is~sufficiently~large}\label{eq:ud_approx}
	\end{align}
\end{Theorem}

Here each $w_n$ follows a binomial distribution of $B(K,P_e)$ because in the uncoded transmission round $r_n$ will miss each data packet with a probability of $P_e$. \eqref{eq:ud_overall} is a straightforward extension of \eqref{eq:ud_a}. We will prove in Appendix \ref{ap:ud_overall} that $E\left[\frac{\sum_{n=1}^N w_n^2}{\sum_{n=1}^N w_n}\right]\approx KP_e-P_e+1$ when $N$ is sufficiently large, which will prove the approximation in \eqref{eq:ud_approx}.

\begin{Remark}
	Under the more general setting with heterogeneous $\{P_{e,n}\}$, we can apply the smallest (resp. largest) $P_{e,n}$ to \eqref{eq:ud_approx} for a lower  (resp. upper) bound on the expected APDD of the perfect LNC technique.
\end{Remark}

To verify the accuracy of the approximation in \eqref{eq:ud_approx} , we simulate the perfect LNC technique as follows: We first send each data packet uncoded once. Then in each time slot, instead of finding a perfect coded packet, we simply give each unfinished receiver one new wanted data packet if its channel is ``on''. The APDD of this simulated technique is thus $\ud$. We have conducted extensive simulations to obtain the numerical average of $\ud$, and compared it with \eqref{eq:ud_approx}. The results under $K=15$ and $P_e=0.2$ are presented in \figref{fig:bounds}, which shows that our approximation quickly converges to the numerical average when $N$ is only 20. Another interesting observation from both the graph and \eqref{eq:ud_approx} is that the APDD of the perfect LNC technique does not increase with the number $N$ of receivers.

\begin{figure}[ht]
	\centering
	\includegraphics[width=0.7\linewidth]{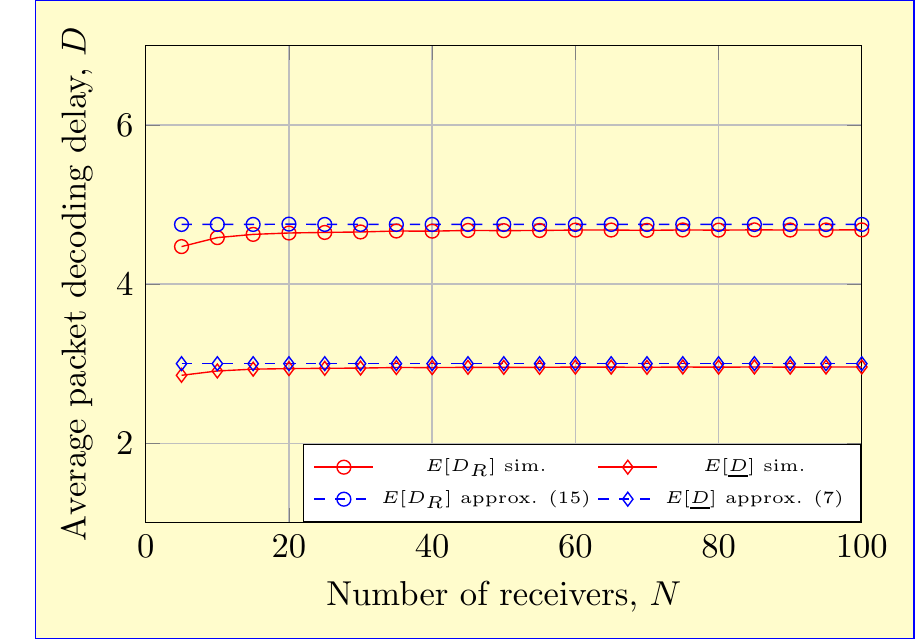}
	\caption{The proposed closed-form approximation of $E[\ud]$ and $E[D_R]$ and their comparison to the simulated values.}
	\label{fig:bounds}
\end{figure}

\subsection{The Hardness of Minimizing the Expected APDD}\label{sec:hardness}
In this section, we prove the NP-hardness of minimizing the expected APDD of a given $\mA$, i.e., $E[D]$. Our approach is through contradiction: If there exists an efficient coding algorithm that minimizes $E[D]$, then this algorithm can efficiently determine the achievability of $E[\ud]$ as well by simply comparing the two quantities. Moreover, such efficiency will hold for the special case where $\{P_{e,n}\}=0$, which will reduce $E[D]$ and $E[\ud]$ to $D$ and $\ud$, respectively. Thus, the achievability of $\ud$ can be efficiently determined. Therefore, to prove that it is hard to minimize $E[D]$, it suffices to prove that it is hard to determine the achievability of $\ud$ when $\{P_{e,n}\}=0$.

{To this end, we define a perfect LNC solution, denoted by $\S_p$, as a set of ordered LNC coded packets that allows every receiver $ r_n$ to decode a new wanted data packet in every transmission, until $ r_n$ has decoded all its wanted data packets. To avoid ambiguity, the coded packets do not contain data packets that have already been decoded by all receivers.}

It is clear that $\ud$ can only be achieved by $\S_p$. Thus, to prove that it is hard to determine the achievability of $\ud$ when $\{P_{e,n}\}=0$, it suffices to prove that it is hard to determine the existence of $\S_p$ when $\{P_{e,n}\}=0$. We now prove this through a reduction to a hypergraph strong coloring problem.

A hypergraph $\mH$ is defined by a pair $(\V,\E)$, where $\V$ is the set of vertices, and $\E$ is the set of hyperedges. Every hyperedge $\e\in\E$ is a subset of $\V$ with size $|\e|\geqslant 1$. A hypergraph is $r-$uniform if every hyperedge $\e\in\E$ has a size of $r$. A size-$k$ strong coloring of $\mH$ is a partition of $\V$ into $k$ subsets $\{\V_i\}_{i=1}^k$, such that $|\V_i\cap\e|\leqslant 1$ for any $\e\in\E$. In other words, if we assign $k$ different colors to the vertices in $\{\V_i\}_{i=1}^k$ respectively, every color appears at most once in every hyperedge.
We prove in Appendix \ref{ap:hardness} that the strong hypergraph coloring problem is hard:

\begin{Lemma}\label{lemma:strong_color}
	It is NP-complete to determine whether an $r-$uniform hypergraph is size-$r$ strong colorable or not, for any $r\geqslant3$.
\end{Lemma}

We then build a reduction from the problem of finding a size-$r$ strong coloring for an $r$-uniform hypergraph to the problem of finding a perfect solution for a given instance of SFM $\mA$. Given an $r$-uniform hypergraph $\mH(\V,\E)$ we construct an instance of our problem as follows. For each vertex $v_k$ we introduce a data packet $\p_k$, and for each hyperedge $\e_n$ we introduce a receiver $ r_n$ that wants the data packets/vertices incident to $\e_n$. 
In the resulting SFM $\mA$, every receiver wants $r$ data packets. A 3-uniform hypergraph $\mH$ and the corresponding SFM $\mA$ are depicted in \figref{fig:sfm_h}.

Bases on this construction, we can prove the hardness of finding a perfect solution:
\begin{Theorem}\label{lemma:det_perfect}
	It is NP-complete to determine whether there exists a perfect solution for a given instance $\mA$ of the APDD minimization problem when there are no packet erasures.
\end{Theorem}

\begin{proof}
	First, we prove that a size-$r$ strong coloring $\{\V_i\}_{i=1}^r$ of $\mH$ implies a perfect solution $\S_p$ for $\mA$ in our problem. Since for every hyperedge it holds that $|\e_n|=r$ and there are $r$ colors, we have $|\V_i\cap\e_n|=1$. Let $\{\M_i\}_{i=1}^r$  be the sets of packets corresponding to $\{\V_i\}_{i=1}^r$. Then, we have \mbox{$|\M_i\cap\w_n|=1$} for every receiver $ r_n$. Hence, the sum (e.g., binary XOR) of all data packets from $\M_i$ is a coded packet, denoted by $\X_i$, that allows every receiver to decode a wanted data packet. Therefore, $\{\X_i\}_{i=1}^r$ together form a perfect solution $\S_p$ to our problem.
	
	Next, we prove that a perfect solution $\S_p$ for the $\mA$ in our problem implies a size-$r$ strong coloring $\{\V_i\}_{i=1}^r$ of $\mH$. Since every receiver wants $r$ data packets, $\S_p$ contains $r$ network coded packets $\{\X_i\}_{i=1}^r$. In order to allow every receiver to decode one wanted data packet from $\X_i$, the coding set $\M_i$ of $\X_i$ must contain exactly one wanted data packet of every receiver, i.e., $|\M_i\cap\w_n|=1$. Let $\{\V_i\}_{i=1}^r$ be the sets of vertices corresponding to $\{\M_i\}_{i=1}^r$. Then it holds that $|\V_i\cap\e_n|=1$ for every hyperedge. Thus, $\{\V_i\}_{i=1}^r$ is a size-$r$ strong coloring of $\mH$.
	
	Our construction above indicates that an $r$-uniform hypergraph is size-$r$ strong colorable if and only if there exists a perfect solution of the instance $\mA$ of our problem. This result, together with Lemma~\ref{lemma:strong_color}, prove the NP-hardness of determining the existence of $\S_p$.
\end{proof}
Since $\ud$ can only be achieved by a perfect solution, the above theorem immediately yields the following corollary:
\begin{Corollary}
	It is NP-compete to determine whether $\ud$ is achievable for a given instance $\mA$ of the APDD minimization problem when there are no packet erasures.
\end{Corollary}
As we have discussed at the beginning of this section, a coding algorithm that can efficiently minimize $E[D]$ can also efficiently determine the achievability of $\ud$ when there are no packet erasures. This relation, together with the above corollary, yields the hardness of APDD minimization problem:
\begin{Theorem}
	It is NP-hard to minimize the expected APDD of a given $\mA$ when coded transmissions are subject to random packet erasures.
\end{Theorem}

Given the NP-hardness of APDD minimization, it becomes critical to investigate whether APDD can be approximation or not, as well as its impacts on throughput optimality. To this end, we study two well known classes of LNC techniques that aim at throughput optimization and APDD reduction, respectively.

\section{Existing Techniques}\label{sec:rlnc_idnc}
In this section, we study the throughput and APDD performance of two classes of LNC techniques, including RLNC, which is throughput optimal, and IDNC, which aims at APDD reduction. We will derive the expected APDD of RLNC, and use the results to prove that RLNC is an APDD 2-approximation technique. On the other hand, we will prove that IDNC techniques cannot approximate throughput, {and can not guarantee an APDD approximation ratio for at least a subset of the receivers.}

\subsection{RLNC}
Since RLNC is throughput optimal and requires block decoding, a receiver who wants $w_n$ data packets is able to decode them altogether after receiving $w_n$ coded packets, which is expected to take $\frac{w_n}{1-P_e}$ coded transmissions. Hence, the expected APDD of $ r_n$ under RLNC is:
\begin{equation}\label{eq:dr_n}
	E[D_{R,n}|w_n]=\frac{w_n}{1-P_{e,n}}
\end{equation}
where the ``R'' in the subscript of $D_{R,n}$ stands for RLNC.

Averaging $E[D_{R,n}|w_n]$ over all receivers yields the expected APDD of a given $\mA$ under RLNC:
\begin{align}\nonumber
	E[D_R|\{w_n\}_{n=1}^N]&=\frac{E[D_{R,1}|w_1]w_1+E[D_{R,2}|w_2]w_2+\cdots+E[D_{R,N}|w_N]}{w_1+w_2+\cdots+w_N}\\
	&=\frac{\sum_{n=1}^N\frac{w_n^2}{1-P_{e,n}}}{\sum_{n=1}^N{w_n}}\label{eq:dr_a}
\end{align}

Then, by assuming that $\mA$ is the consequence of one round of uncoded transmissions of the $K$ data packets, we obtain the following theorem on the overall expected APDD performance of RLNC:
\begin{Theorem}\label{theo:rd_overall}
	Given system parameters $K$, $N$, and $P_e$, the overall APDD performance of RLNC is:
	\begin{align}
		E[D_R]
		&=\frac{1}{1-P_e}E\left[\frac{\sum_{n=1}^N w_n^2}{\sum_{n=1}^N w_n}\right],~~~~\{w_n\}_{n=1}^N\sim B(K,P_e)\label{eq:ur_overall}\\
		&\approx\frac{KP_e-P_e+1}{1-P_e},~~~~\mathrm{when~}N~\mathrm{is~sufficiently~large}\label{eq:dr_approx}
	\end{align}
\end{Theorem}
The accuracy of our approximation is again confirmed by the simulation results plotted in \figref{fig:bounds}.

By comparing the above 3 expected APDD under RLNC with the corresponding lower bounds (\eqref {eq:ud_n},  \eqref{eq:ud_a}, \eqref{eq:ud_approx}), we can easily verify that the ratio between the expected APDD of using RLNC and the lower bound converges to 2 from below with increasing $w_n$ or $K$. Thus, the expected APDD of using RLNC is at most twice of the minimum expected APDD. We thus have the following theorem:
\begin{Theorem}\label{theo:rlnc}
	RLNC is an APDD 2-approximation technique when coded transmissions are subject to random packet erasures.
\end{Theorem}
%
Since RLNC requires block decoding, any throughput optimal LNC techniques that can provide early packet decodings will outperform RLNC in terms of APDD. Thus, all LNC techniques can approximation APDD with a ratio of at most 2. Next, we will raise an example in which this ratio is 4/3. These two results yields the following bounds on the approximation ratio of throughput-optimal LNC techniques:

\begin{Theorem}\label{theo:optimal}
	The APDD approximation ratio of all throughput-optimal LNC techniques is between $\frac{4}{3}$ and 2.
\end{Theorem}

\begin{proof}
	Since approximation ratio is the highest ratio among all settings, to prove that $\beta\geqslant \frac{4}{3}$, it suffices to show an instance in which the APDD of all throughput-optimal LNC techniques is $\frac{4}{3}$ times of the minimum.
	
	Consider an instance $\mA$ of the APDD minimization problem with $2$ data packets and $N$ receivers. Receiver $ r_1$ only wants $\p_1$, receiver $ r_2$ only wants $\p_2$, and all the remaining $N-2$ receivers want both $\p_1$ and $\p_2$. Further assume that $\{P_{e,n}\}_{n=1}^N=0$, so that $E[D]=D$. When $N$ is sufficiently large, the APDD is minimized if $\p_1$ and $\p_2$ are transmitted separately using two transmissions. This yields $\Dmin=1.5$. On the other hand, if throughput-optimal techniques are applied, $\p_1$ and $\p_2$ must be combined in the first transmission due to $ r_1$ and $ r_2$, which does not allow the remaining $N-2$ receivers to decode in the first transmission. These receivers can only decode in the second transmission. Consequently, the APDD is 2 when $N$ is sufficiently large, which is $\frac{4}{3}$ times of the minimum.
\end{proof}

\subsection{IDNC}

A common feature of IDNC techniques, including strict IDNC (S-IDNC) \cite{Rozner_Heuristic_clique,sundararajan:sadeghi:medard:2009,yu2015SIDNC} and general IDNC (G-IDNC) \cite{sorour:valaee:2010,li2011capacity,le2013instantly,sorour2015completion}, is that every coded packet is the binary XOR of a subset of the data packets. Another common feature of IDNC techniques is ``memoryless decoding'':  the receivers discard the coded packets that are innovative but not instantly decodable rather than storing them for future decodings\footnote{{There are recent works in IDNC literature that allow the receiver to store at least one innovative but not instantly decodable packet \cite{neda:parastoo:o2idnc}.}}. These features reduce decoding complexity, and motivates the sender to send coded packets that are instantly decodable to as many receivers as possible. Such decoding helps reduce APDD, but at a cost of an unbounded degradation of the throughput:

\begin{Theorem}\label{theo:mem_weak_thpt}
	IDNC techniques are not throughput-approximation techniques.
\end{Theorem}

\begin{proof}
	If IDNC techniques are throughput-approximation techniques, then when there are not packet erasures, the BCT $U_I$ provided by IDNC techniques should be within a constant multiple of the minimum BCT $\Umin$ that any LNC techniques can offer. Thus, to prove the theorem, we only need a counter example in which $U_I$ is not within a constant multiple of $\Umin$.
	
	Our counter example involves two types of SFMs:
	\begin{itemize}
		\item $\mA_1(K)$: every pair of data packets is wanted by a different receiver. Thus, there are $N=\frac{K(K-1)}{2}$ receivers in total;
		\item $\mA_2(K,m)$: every data packet is wanted by $m$ different receivers. Every pair of data packets is wanted by a different receiver. Thus, there are $N=mK+\frac{K(K-1)}{2}$ receivers in total.
	\end{itemize}
	
	Note that $\Umin=2$ for $\mA_1(K)$, which can be achieved by any throughput-optimal LNC techniques. We prove the theorem by proving that $U_I=\left\lceil\log_2K\right\rceil+1$ for $\mA_1(K)$, where $\lceil x \rceil$ denotes the smallest integer greater than $x$.
	
	Given $\mA_1(K)$, the transmission starts by sending as $\cc_1$ the XOR of an arbitrary $m_1 \geqslant 1$ data packets in $\mA_1(K)$. The resulted SFM consists of two sub-SFMs: 1) an $\mA_1(m_1)$, which contains the $m_1$ data packets and the receivers who want 2 data packets from $\cc_1$ and thus discard $\cc_1$; 2) an $\mA_2(K-m_1, m_1)$, which contains the remaining $K-m_1$ data packets and the remaining receivers, which either has decoded one wanted data packet from $\cc_1$ and still wants one data packet from $K-m_1$, or wants 2 data packets from $K-m_1$.
	
	These two sub-SFMs are independent in the sense that an IDNC coded packet of $\mA_1(m_1)$ and an IDNC coded packet of $\mA_2(K-m_1,m_1)$ can be XOR-ed and sent without affecting their decodability at their targeted receivers.
	
	Similarly, we can show that after sending the XOR of any arbitrary $m_3$ data packets from $A_2(K-m_1, m_1)$, the resulted SFM consists of two independent sub-SFMs: an $\mA_1(m_3)$ and an $\mA_2(K-m_1-m_3,m_1+m_3)$.
	
	Continuing the logic, after the $u$-th transmission ($u\geqslant 1$), $\mA_1(K)$ is split into $2^{u-1}$ type-1 sub-SFMs and $2^{u-1}$ type-2 sub-SFMs. Only sub-SFMs that consists of a single data packet can be completed in one coded transmission and be removed. The evolution of $\mA_1(K)$ is demonstrated in a layered graph in Fig. \ref{fig:memoryless_decoding}. The $u$-th layer corresponds to the SFM before the $u$-th coded transmission. The total number of coded transmissions is thus the number of layers plus one. It is clear that the minimum number of layers is $\left\lceil\log_2K\right\rceil$, which is achieved by XOR-ing half of the data packets from each sub-SFM. Thus, $U_I=\left\lceil\log_2K\right\rceil+1$.
\end{proof}

\begin{figure}[t]
	\includegraphics[width=\linewidth]{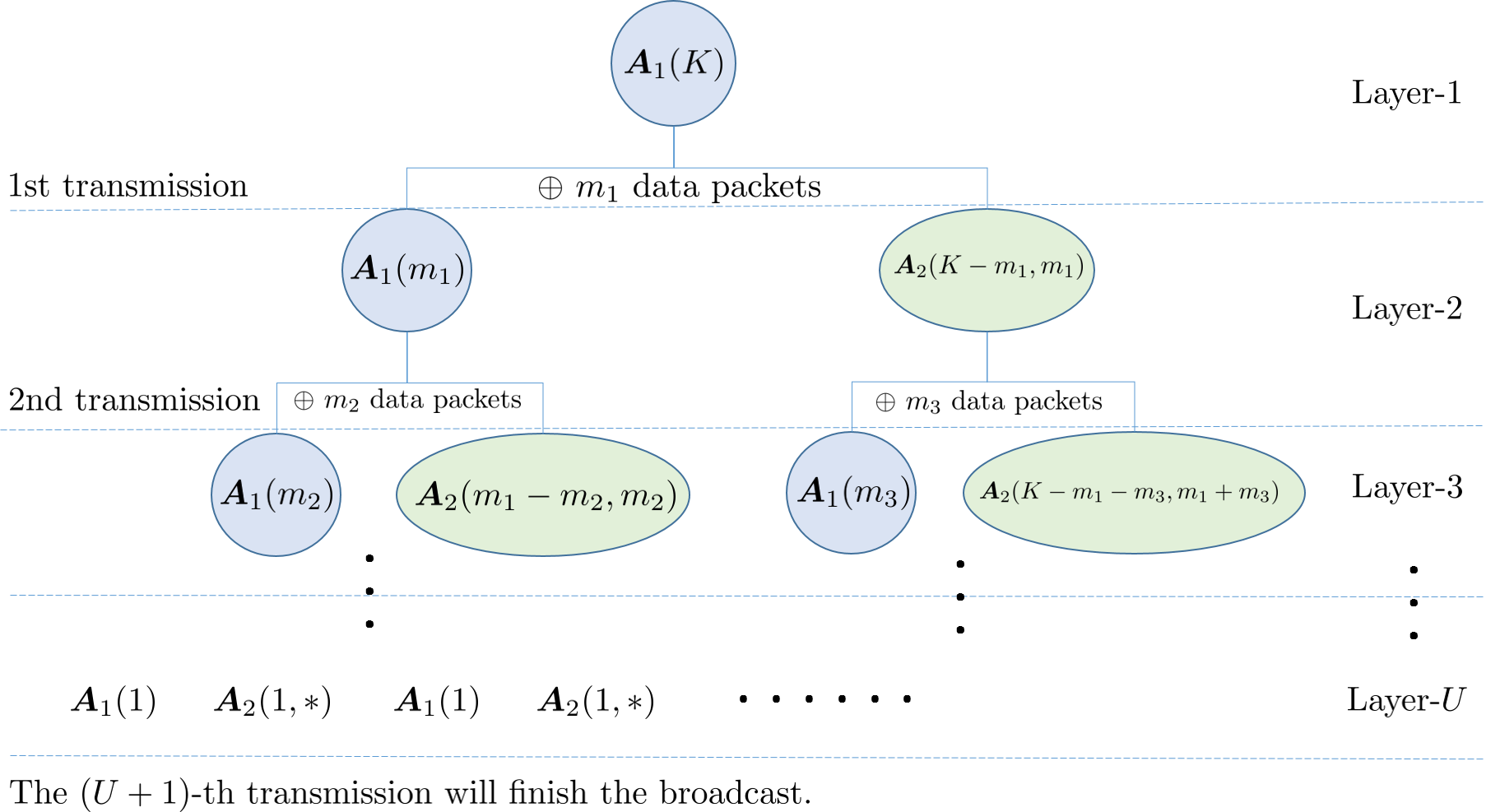}
	\caption{The decoding evolution of $\mA_1(K)$. (Note that '*'s are arbitrary positive integers whose values depend on the coded packets.)}
	\label{fig:memoryless_decoding}
\end{figure}

Our theorem indicates that all IDNC techniques are only heuristics for BCT minimization and throughput maximization. {A large BCT can, in turn, hurts APDD, as there are data packets decoded with large delays. In the example raised in the above proof, when IDNC techniques are applied, there are receivers who decode their second wanted data packet after $\left\lceil\log_2K\right\rceil+1$ transmissions. Even if they decode their first wanted packet after the first transmission, their APDD is still as high as $\left\lceil\log_2K\right\rceil/2+1$, which can be unlimitedly larger than an APDD of 2 offered by RLNC. Therefore, IDNC techniques cannot approximate the APDD of at least a subset of the receivers. However, it is still an open problem whether IDNC techniques can approximate the APDD across all receivers.}

\section{An Improved APDD Approximation LNC Technique}\label{sec:framework}
In this section, we propose hypergraphic LNC (HLNC), a novel low complexity throughput-optimal and APDD-approximation technique built upon a novel hypergraph model of receivers' packet reception state. Coded packets generated by HLNC have the following two features: 
\begin{enumerate}
	\item every coded packet is innovative to every receiver; and
	\item every coded packet is able to offer at least one early packet decoding.
\end{enumerate}

The first feature ensures the throughput-optimality and the APDD-approximation of HLNC. The second feature further ensures that both the APDD and decoding complexity of HLNC are lower than RLNC, as RLNC requires block decoding.

\subsection{The Basic Form of HLNC}
We introduce our technique by first generalizing the concept of \emph{Wants} set $\W_n$: 
\begin{Definition}
	The Wants set $\W_n$ of a receiver $ r_n$ is the set of data packets not yet decoded by $ r_n$.
\end{Definition}
The subtle yet important difference between this new definition and the previous one is that the new one explicitly includes in $\W_n$ the received (from coded packets) but undecodable data packets.

We then model the packet reception state $\{\W_n\}$ using the following hypergraph.
\begin{Definition}\label{def:H_new}
	In the hypergraph model $\mH(\V,\E)$ of the receivers' packet reception state, each vertex $v\in\V$ represents a data packet, i.e., $v_k\leftrightarrow\p_k$.  Each hyperedge $\e\in\E$ represents the Wants set of a receiver, i.e, $\e_n\leftrightarrow \W_n$, by connecting the data packets/vertices in $\W_n$.
\end{Definition}

An example of $\mH$ is plotted in \figref{fig:h_color}, where 6 data packets and 3 receivers are modelled as a hypergraph with 6 vertices and 3 hyperedges. 

We further denote by $\V_c$ a minimal vertex cover of $\mH$, and denote the corresponding packet set as $\M_c$. $\V_c$ is a subset of $\V$ satisfying that 1) it is incident to every hyperedge, i.e., $|\e_n\cap\V_c|>0~\forall n\in[1,N]$, and 2) there is a least one single incidence, i.e. $\exists n:|\e_n\cap\V_c|=1$. Due to these two features, a properly generated (will be discussed soon) coded packet of $\M_c$ has the following two properties: 

\begin{enumerate}
	\item It is innovative to every receiver, because every receiver wants at least one data packet from $\M_c$, i.e., $|\W_n\cap\M_c|>0~\forall n\in[1,N]$; 
	\item It allows at least one receiver to instantly decodable one wanted data packet, because there is at least one receiver who only wants one data packet from $\M_c$, i.e., $\exists n:|\W_n\cap\M_c|=1$.
\end{enumerate}

Therefore, the core of HLNC is to keep updating $\mH$ and sending coded packets generated using $\V_c$. The basic HLNC technique is outlined in Algorithm \ref{alg:hlnc}. A complete example of the basic HLNC is given in Example \ref{example:hlnc} at the end of this section.
\begin{algorithm}[h]
	\caption{HLNC Wireless Broadcast}
	\begin{algorithmic}[1]
		\STATE Input: the initial packet reception state $\mH$;
		\WHILE  {Not all receivers have decoded all wanted data packets} 
		\STATE The sender updates $\mH$;
		\STATE The sender broadcasts a properly generated coded packet $\cc$ of a minimal vertex cover $\V_c$ of $\mH$;
		\STATE Every receiver that receives this coded packet tries to decode by solving linear equations(s).
		\ENDWHILE
	\end{algorithmic}
	\label{alg:hlnc}
\end{algorithm}

Coded packets must be properly generated to ensure their innovativeness to every receiver. This can be easily accomplished either asymptotically, by random coefficients chosen from a large $\F_q$ (as RLNC), or deterministically, by coefficients generated using deterministic LNC techniques such as \cite{kwan2011generation}\footnote{We note that the coding set selection strategies of these techniques are completely different from HLNC. For example, RLNC selects all the data packets, whilst the technique in \cite{kwan2011generation} selects through solving an NP-hard hitting set problem}.

On the other hand, a minimal vertex cover can be found in polynomial-time by simple heuristic algorithms. We propose  such an algorithm in Algorithm \ref{alg:dmin}, which prioritizes data packets based on their popularity.

{One way to understand HLNC is that HLNC aims at achieving throughput optimality by encoding the minimum number of data packets together. Such minimization will minimize the intersection between the coding set and receivers' Wants sets to a minimum of one data packet, which makes such data packets instantly decodable to the corresponding receivers. We note that although \cite{keller2008online} has a similar aim, the way how the coding set is determined is completely different, and there is no claim on guaranteed instantly packet decodings nor guaranteed APDD approximation.}

\begin{figure}
	\centering
	\subfigure[SFM $\mA$]{\includegraphics[width=0.25\linewidth]{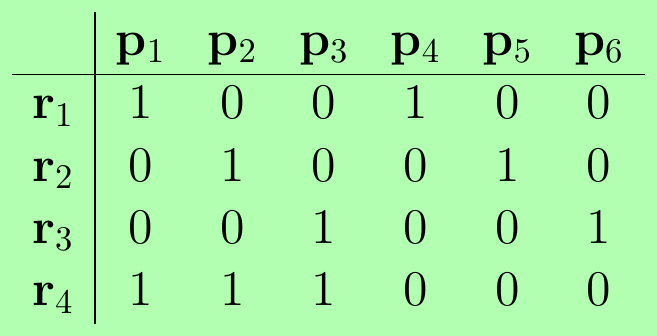}\label{fig:sfm1}}
	\subfigure[Original $\mH$]{\includegraphics[width=0.2\linewidth]{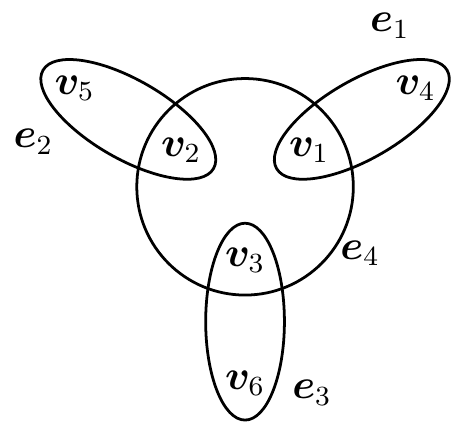}\label{fig:h1}}\hspace{10pt}
	\subfigure[updated $\mH'$ after sending $\a_1\p_1+\a_2\p_2+\a_3\p_3$]{\includegraphics[width=0.2\linewidth]{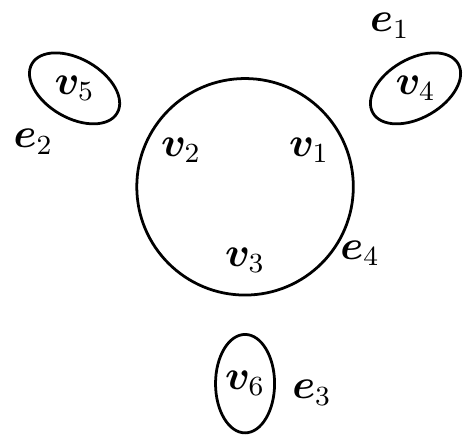}\label{fig:h2}}\hspace{10pt}
	\subfigure[updated $\mH''$ after sending $\a'_1\p_1+\a'_4\p_4+\a'_5\p_5+\a'_6\p_6$]{\includegraphics[width=0.15\linewidth]{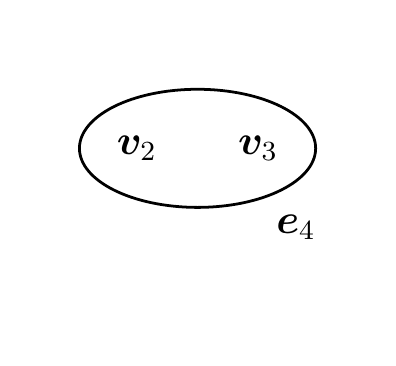}\label{fig:h3}}
	\caption{Hypergraph update}
\end{figure}
\begin{algorithm}[t]
	\caption{Find a minimal hypergraph vertex cover}
	\begin{algorithmic}[1]
		\STATE Input: A hypergraph $\mH(\V,\E)$, an empty vertex set $\V_c$;
		\STATE Weigh each vertex with the number of hyperedges incident to it; (\emph{The weight of a vertex is indeed the number of receivers which want the corresponding data packet.})
		\WHILE {there are still vertices in $\mH$}
		\STATE Add to $\V_c$ the vertex with the largest weight;
		\STATE Update $\mH$ by: 1) removing this vertex from $\mH$; 2) removing from $\mH$ all hyperedges incident to this vertex; 3) removing from $\mH$ all vertices that do not have any hyperedge incident to them;
		\ENDWHILE
		\STATE Output $\V_c$ as a minimal hypergraph vertex cover.
	\end{algorithmic}
	\label{alg:dmin}
\end{algorithm}

\begin{Example}\label{example:hlnc}
	Consider the SFM in \figref{fig:sfm1}, with its hypergraph $\mH$ plotted in \figref{fig:h1}. For simplicity we assume erasure-free transmissions. Then, the transmissions using HLNC is as follows:
	\begin{enumerate}
		\item A minimal vertex cover of $\mH$ is $\V_c=\{v_1,v_2,v_3\}$. Thus, $\X=\alpha_1\p_1+\alpha_2\p_2+\alpha_3\p_3$ is sent, where $\{\a_1,\a_2,\a_3\}$ are coefficients chosen from $\F_q\setminus\{0\}$. Since receivers $ r_1,~ r_2,~ r_3$ only want $\p_1,~\p_2,~\p_3$ from $\X$, respectively, they can decode them from $\X$. Hence, $v_1,~v_2,~v_3$ are removed from $\e_1,~\e_2,~\e_3$, respectively. On the other hand, $ r_4$ cannot decode any data packet. It only holds $\X$, and thus $\e_4$ is still incident to $\{v_1,v_2,v_3\}$. The updated graph $\mH'$ is plotted in \figref{fig:h2};
		\item A minimal vertex cover of $\mH'$ is $\V_c'=\{v_1,v_4,v_5,v_6\}$. Thus, $\X'=\alpha'_1\p_1+\alpha'_4\p_4+\alpha'_5\p_5+\alpha'_6\p_6$ is sent. Since receivers $ r_1,~ r_2,~ r_3,~ r_4$ only want $\p_4,~\p_5,~\p_6,~\p_1$ from $\X'$, respectively, they can decode them from $\X'$. Hence, $v_4,~v_5,~v_6,~v_1$ are removed from $\e_1,~\e_2,~\e_3,~\e_4$, respectively. Since receivers $\{ r_1, r_2, r_3\}$ are satisfied, $\{\e_1,\e_2,\e_3\}$ are completely removed from $\mH$. On the other hand, $ r_4$ holds a linear equation of $\alpha_2\p_2+\alpha_3\p_3=\X-\alpha_1\p_1$, where the value of the RHS is known to $ r_4$. Since $ r_4$ still wants $\{\p_2,\p_3\}$, $\e_4$ is incident to $\{v_2,v_3\}$. The updated graph $\mH''$ is plotted in \figref{fig:h3}.
		\item A minimal vertex cover of $\mH''$ is $v_2$. Thus, $\p_2$ is sent alone, which will allow $ r_4$ to obtain $\p_2$, and then decode $\p_3$ from equation $\alpha_2\p_2+\alpha_3\p_3=\X-\alpha_1\p_1$. The coded broadcast will then be completed.
	\end{enumerate}
\end{Example}


%
%
\subsection{Implementation under Different Feedback Availability}
The basic HLNC proposed in Algorithm \ref{alg:hlnc} works when the sender can update $\mH$ after every transmission. This requires fully-online receiver feedback, namely every receiver sends a feedback every time it receives a coded packet. However, it could be expensive or even impossible to collect feedback after every transmission. For example, in time-division-duplex systems the sender has to stop and listen to the feedback \cite{lucani:tdd_rlnc:2009,lucani:tdd_field:2009,lucani:medard:stojanovic:2010,lucani:medard:stojanovic:2012}. Thus, it is desirable to develop implementations of HLNC with semi-online feedback and without intermediate feedback at all (off-line). The challenge here is how to update $\mH$ locally without losing throughput-optimality and APDD-approximation.

We solve this challenge by our semi-online HLNC proposed in Algorithm \ref{alg:hlnc_semi}. At a high level, we partition the transmissions into semi-online rounds. During each semi-online round (Step 5-9 in Algorithm \ref{alg:hlnc_semi}), the sender iteratively broadcasts a coded packet $\cc$ and then update $\mH$ locally by assuming $\cc$ is received by all receivers. The iteration will be terminated if the reception of $\cc$ will allow one receiver to fully decode all its wanted data packets. After the termination, the sender will collect receiver feedback to correctly update $\mH$, and then enter the next semi-online round.

Due to this termination condition, the locally updated $\mH$ always has the same number of hyperedges/receivers as the actual packet reception state $\mH'$ at the receivers with. Moreover, since we update $\mH$ by assuming all coded packets are received by all receivers, the vertices removed from $\mH$ during the updates are at least as many as $\mH'$. Thus, $\mH$ is also a subgraph of $\mH'$. Consequently, $\mH$ and $\mH'$ have the following easily-proved property:
\begin{property}
	If $\mH$ is a subgraph of $\mH'$ with the same number of hyperedges, then any minimal vertex cover $\V_c$ of $\mH$ is also a minimal vertex cover of $\mH'$.
	\label{prop:mvc_sub}
\end{property}

This property immediately indicates that the coded packets generated by semi-online HLNC have the same features as those generated by the fully-online basic HLNC, namely, 1) innovative to every receiver, and 2) always provide early packet decodings. This is because the minimal vertex covers we found in each semi-online round are also minimal vertex covers of the actual $\mH'$ at the receivers.

Therefore, semi-online HLNC is throughput-optimal and APDD-approximating as its fully-online counterpart. As we will see in the simulation results presented in the next section, there is no visible performance difference between the two.

Based on semi-online HLNC, we can also develop an offline HLNC by not collecting receiver feedback after the first semi-online round, but using the classic RLNC in the remaining transmissions. As we will see in the simulation results presented in the next section, even off-line HLNC can always provide lower APDD than RLNC.

{Moreover, HLNC is very robust to feedback loss. If the feedback from a subset of receivers is lost, the sender can simply update the hypergraph by assuming that these receivers have not received any previous packets.
	
	Last but not least, HLNC is fully compatible with RLNC and offers seemless switch. The sender can keep sending RLNC packets until one round of receiver feedback is received. Based on the feedback, the sender conducts one HLNC semi-online round, and then switches back to RLNC until the reception of the next round of receiver feedback.}

\begin{algorithm}[h]
	\caption{HLNC Wireless Broadcast with Semi-online Feedback}
	\begin{algorithmic}[1]
		\STATE Input: the initial packet reception state $\mH$;
		\WHILE  {Not all receivers have decoded all wanted data packets} 
		\STATE The sender updates $\mH$ by collecting receiver feedback (waived for the first transmission);
		\STATE The sender broadcasts a properly generated coded packet $\cc$ of a $\V_c$ of $\mH$;
		\WHILE{No receiver can fully decode all its wanted data packets after receiving $\cc$}
		\STATE The sender updates $\mH$ locally by assuming $\cc$ is received by all receivers;
		\STATE The sender broadcasts a properly generated coded packet $\cc$ of a $\V_c$ of $\mH$;
		\STATE Every receiver that receives this coded packet tries to decode by solving linear equations(s).
		\ENDWHILE
		\ENDWHILE
	\end{algorithmic}
	\label{alg:hlnc_semi}
\end{algorithm}

\section{Simulation Results}\label{sec:simulation}

In this section, we numerically compare the APDD performance of the proposed technique with some existing techniques, as well as the lower bound on APDD. In total, there are 6 different APDD we will compare. They are abbreviated and explained as follows:
\begin{enumerate}
	\item ``Fully-'': the  APDD of fully-online HLNC;
	\item ``Semi-'': the APDD of semi-online HLNC;
	\item ``Off-'': the APDD of offline HLNC;
	\item ``$E[\ud$]'': the lower bound on the expected APDD of LNC derived in \eqref{eq:ud_approx};
	\item ``$E[D_R]$'': the expected APDD of RLNC derived in \eqref{eq:dr_approx};
	\item ``G-IDNC'': the APDD of a heuristic GIDNC algorithm adapted from \cite{sorour:valaee:2010} when fully-online feedback is collected. We note that there has not been any optimal G-IDNC algorithms in the literature.
\end{enumerate}
In our simulations, there are $K=15$ data packets, $N\in[5,100]$ receivers. The packet erasure probabilities are $\{P_{e,n}\}_{n=1}^N=0.2$. For each value of $N$, we simulate the broadcast of $10^5$ packet blocks, and then make average on their APDD. The simulation results are plotted in \figref{fig:d_erasure}, from which we observe that:
\begin{itemize}
	\item The APDD performance of our technique outperforms the existing techniques. This superiority holds even when our technique is implemented under the off-line scheme;
	\item The fully- and semi-online schemes share the same performance. This result matches our expectation. Their performance is better than the off-line one;
	\item The APDD of RLNC is always within a constant factor of the lower bound, indicating that RLNC is an approximation technique, and so is our technique. On the other hand, the APDD of the heuristic G-IDNC is unbounded, indicating that it is not an approximation technique. We further note that one should not read the approximation ratio from the figure, because the approximation ratio is achieved when $K$ is sufficiently large, as have discussed in Section \ref{sec:rlnc_idnc} before Theorem \ref{theo:rlnc}.
\end{itemize}
\begin{figure}[ht]
	\centering
	\includegraphics[width=0.7\linewidth]{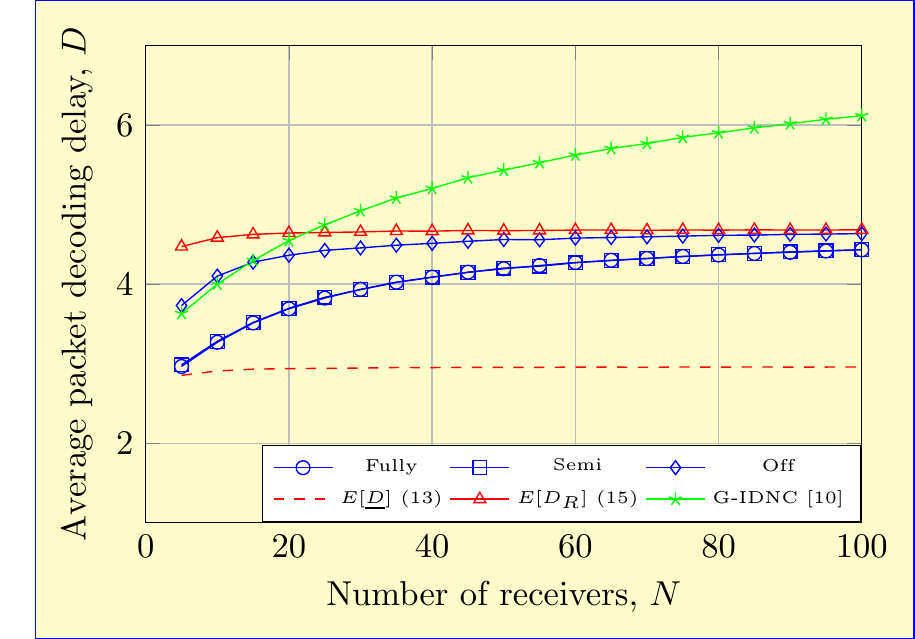}
	\caption{The APDD performance of different NC techniques.}
	\label{fig:d_erasure}
\end{figure}

In addition, we compare the amount of feedback collected under fully- and semi-online HLNC. The results are plotted in \figref{fig:feedback}. We observe that semi-online HLNC can reduce up to 30\% feedback from the fully-online one when the number of receivers is small. The reduction becomes marginal with increasing number of receivers because the probability of having a receiver who only wants one data packet after a certain semi-online round increases. When this happens, semi-online HLNC has to collect feedback after only one transmission according to Algorithm \ref{alg:hlnc_semi}, which makes it equivalent to fully-online HLNC. 

\begin{figure}[ht]
	\centering
	\includegraphics[width=0.7\linewidth]{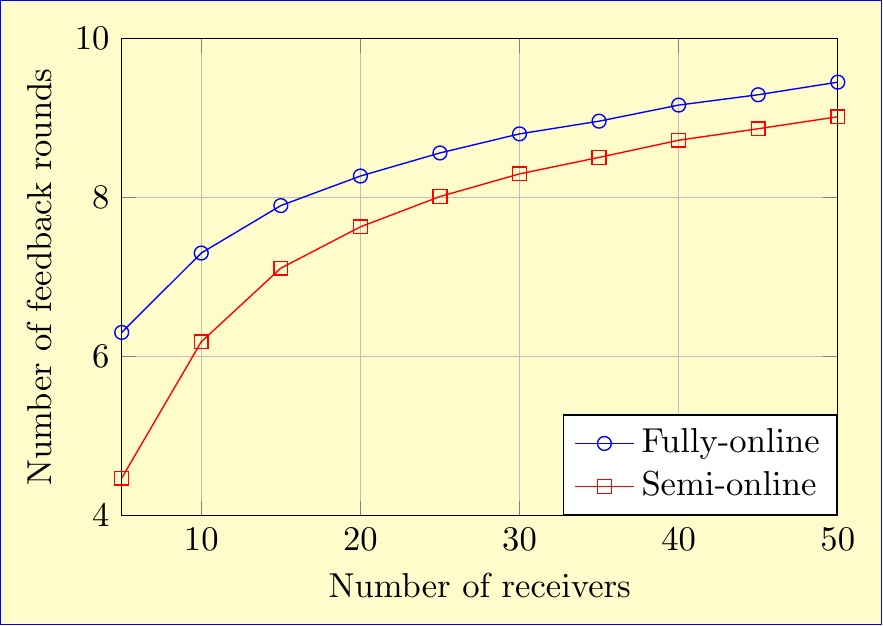}
	\caption{The Amount of feedback collected under the fully- and semi-online schemes.}
	\label{fig:feedback}
\end{figure}

\section{Conclusion}
In this paper, we conducted a comprehensive study on the APDD performance of LNC techniques in wireless broadcast with packet erasures. By deriving lower bounds on the expected APDD of LNC techniques using a conceptual perfect LNC technique, we showed that the APDD of efficient LNC techniques should not scale up with increasing number of receivers. Then by proving the NP-hardness of achieving these lower bounds, we proved the NP-hardness of the APDD minimization problem.

Although the optimal APDD is intractable to achieve, we proved that it can be approximated with a ratio between 4/3 and 2 by throughput-optimal LNC techniques. Therefore, throughput and APDD can be jointly improved rather than trading each other off. However, such joint improvement does not necessarily hold for every LNC technique. {For example, we proved that all IDNC techniques cannot approximate the optimal throughput, and such sub-optimality in turn makes IDNC techniques unable to guarantee an APDD approximation ratio for at least a subset of receivers.}

We then proposed HLNC, a novel hypergraph-based LNC technique that combines all the advantages of RLNC and IDNC: It is throughput-optimal, APDD-approximating, computational friendly, and it always provides instant packet decodings. Moreover, its performance does not degrade even if receiver feedback is not collected after every transmission. Our extensive simulations showed that the APDD of HLNC outperforms RLNC and a heuristic general IDNC technique under all tested system parameter settings.

In the future, we are interested in extending our hypergraph model and APDD analysis to other network coding and index coding problems. We are also interested in applying our technique to other network models such as cooperative data exchange \cite{el2010coding,sprintson2010randomized}.

\appendices
\section{Proof of Lemma \ref{lemma:strong_color}}\label{ap:hardness}

We prove that it is NP-complete to determine whether an $r-$uniforma hypergraph is size-$r$ strong colorable or not, for any $r\geqslant 3$. Our method is a reduction from the $k$-coloring problem of graphs.

Given an arbitrary graph $\G(\V,\E)$. For every edge $\e_n=(v_i,v_j)$ we construct a hyperedge $\e_n'$ by adding $r-2$ dummy vertices $v^{n,1}\cdots v^{n,r-2}$ to $\e_n$. The result is an $r$-uniform hypergraph $\mH(\V',\E')$ that has $|\V'|=|\V|+|\E|\cdot (r-2)$ vertices. If $\G$ can be colored using $r$ colors, then in any hyperedge $\e_n'=(v_i,v_j,v^{n,1}\cdots v^{n,r-2})$, $v_i$ and $v_j$ are colored differently using 2 colors. By assigning the remaining $r-2$ colors to the $r-2$ dummy vertices in $\e_n'$ greedily, all vertices in $\e_n'$ are colored differently. We thus obtain a size-$r$ strong coloring of $\mH$. On the other hand, if $\mH$ can be strong colored using $r$ colors, then by removing all the dummy vertices, we obtain an $r$-coloring of $\G$. It is well known that it is NP-complete to determine whether a graph is $r$ colorable or not, for any $r\geqslant 3$. Hence, it is NP-complete to determine whether an  $r$-uniform hypergraph is size-$r$ strong colorable. $\blacksquare$

\section{Proof of Theorem \ref{theo:ud_n}}\label{ap:ud_n}
We first assume that $ r_n$ is satisfied (i.e., obtains all data packets) after $U_n$ coded transmissions. Hence the packet reception pattern $\bu_n$ of $ r_n$ takes a form of $[u_1,\cdots,u_{w_n-1},U_n]$. Let $\bu'=[u_1,\cdots,u_{w_n-1}]$, then the APDD under a given
$\bf u$ is $(\|{\bf u'}\|+U_n)/w_n$. It is intuitive that
all the possible $\bf u'$ happens with the same probability. Denote the set of all possible $\bf u'$ by $\mathcal U$,
then the expected lower bound under given $\w_n$ and $U_n$, denoted by $E[\ud_n|(U_n,w_n)]$, is calculated as:
\begin{align}\nonumber
	E[\ud_n|(U_n,w_n)]&=\frac{1}{|\mathcal U|}\sum_{{\bf u'}\in\mathcal U}\frac{\|{\bf u'}\|+ U_n}{w_n}\\
	&=\frac{U_n}{w_n}+\frac{1}{|\mathcal U|}\sum_{{\bf u'}\in\mathcal U}\frac{\|{\bf u'}\|}{w_n}\label{eq:ud_single}
\end{align}

We now show that all $\bf u'$ are symmetric. Given ${\bf u'}=[u_1,\cdots,u_{w_n-1}]$, by letting $u''_i=U_n-u_i$,
the resultant ${\bf u''}=[u''_1,\cdots,u''_{w-1}]$ is the mirror of $\bf u'$ against $U_n/2$. Obviously, $\bu''$ also belongs to $\mathcal U$, and it holds that ${\bf \|{\bf u'}\|+\|{ u''}\|}=(w_n-1)U_n$.
Hence, there are $|\mathcal U|/2$ such pairs, and thus the above equation can be simplified to:
\begin{equation}
	E[\ud_n|(U_n,w_n)]=\frac{U_n}{w_n}+\frac{1}{2}\frac{(w_n-1)U_n}{w_n}=\frac{U_n}{2}+\frac{U_n}{2w_n}
\end{equation}

Then, by noting that the expected number of coded transmissions for a receiver to be satisfied is $E[U_n|w_n]=w_n/(1-P_e)$, we have:
\begin{equation}
	E[\ud_n|w_n]=\frac{E[U_n|w_n]}{2}+\frac{E[U_n|w_n]}{2w_n}=\frac{w_n+1}{2(1-P_e)}
\end{equation}
which proves Theorem \ref{theo:ud_n}.
$\blacksquare$

\section{Proof of Theorem \ref{theo:ud_overall}}\label{ap:ud_overall}
Here we only need to prove that $E\left[\frac{\sum_{n=1}^N w_n^2}{\sum_{n=1}^N w_n}\right]\approx KP_e-P_e+1$ when $N$ is sufficiently large. We first expand $E\left[\frac{\sum_{n=1}^N w_n^2}{\sum_{n=1}^N w_n}\right]$ into its series form:

\begin{equation}
	E\left[\frac{\sum_{n=1}^N w_n^2}{\sum_{n=1}^N w_n}\right]=E\left[\frac{w_1^2}{\sum_{n=1}^N w_n}\right]+E\left[\frac{w_2^2}{\sum_{n=1}^N w_n}\right]+\cdots+E\left[\frac{w_N^2}{\sum_{n=1}^N w_n}\right]
\end{equation}
Since $\{w_n\}_{n=1}^N$ are i.i.d. distributed, the $N$ addends in the above equation have the same value. Thus,

\begin{align}
	E\left[\frac{\sum_{n=1}^N w_n^2}{\sum_{n=1}^N w_n}\right]&=N\cdot E\left[\frac{w_1^2}{\sum_{n=1}^N w_n}\right]\\
	&=N\cdot E\left[\frac{w_1^2}{w_1+\sum_{n=2}^Nw_n}\right]
\end{align}

Then according to the law of larger numbers, the value of $\sum_{n=2}^Nw_n$ will approach to its mean $(N-1)KP_e$ when $N$ is sufficiently large. Thus,
\begin{align}
	E\left[\frac{\sum_{n=1}^N w_n^2}{\sum_{n=1}^N w_n}\right]
	&\approx N\cdot E\left[\frac{w_1^2}{w_1+(N-1)KP_e)}\right],~~~~\mathrm{when~}N~\mathrm{is~sufficiently~large}\\
	&=E\left[\frac{w_1^2}{KP_e+\frac{w_1-KP_e}{N}}\right]\\
	&\approx E\left[\frac{w_1^2}{KP_e}\right],~~~~\mathrm{when~}N~\mathrm{is~sufficiently~large}
\end{align}

Then, since $w_1\sim B(K,P_e)$, we have $E[w_1]=KP_e$ and $Var[w_1]=KP_e-KP_e^2$. Hence, we have $E[w_1^2]=E[w_1]^2+Var[w_1]=K^2P_e^2+KP_e-KP_e^2$, and thus $E\left[\frac{\sum_{n=1}^N w_n^2}{\sum_{n=1}^N w_n}\right]\approx KP_e-P_e+1$. $\blacksquare$
\bibliographystyle{IEEEtran}
\bibliography{IEEEabrv,My_ref}

\begin{thebibliography}{10}
\providecommand{\url}[1]{#1}
\csname url@samestyle\endcsname
\providecommand{\newblock}{\relax}
\providecommand{\bibinfo}[2]{#2}
\providecommand{\BIBentrySTDinterwordspacing}{\spaceskip=0pt\relax}
\providecommand{\BIBentryALTinterwordstretchfactor}{4}
\providecommand{\BIBentryALTinterwordspacing}{\spaceskip=\fontdimen2\font plus
\BIBentryALTinterwordstretchfactor\fontdimen3\font minus
  \fontdimen4\font\relax}
\providecommand{\BIBforeignlanguage}[2]{{%
\expandafter\ifx\csname l@#1\endcsname\relax
\typeout{** WARNING: IEEEtran.bst: No hyphenation pattern has been}%
\typeout{** loaded for the language `#1'. Using the pattern for}%
\typeout{** the default language instead.}%
\else
\language=\csname l@#1\endcsname
\fi
#2}}
\providecommand{\BIBdecl}{\relax}
\BIBdecl

\bibitem{Yeung_flow}
R.~Ahlswede, N.~Cai, S.~Li, and R.~Yeung, ``Network information flow,''
  \emph{{IEEE} Trans. Inf. Theory}, vol.~46, no.~4, pp. 1204--1216, Jul. 2000.

\bibitem{li2003linear}
S.-Y. Li, R.~W. Yeung, and N.~Cai, ``Linear network coding,'' \emph{IEEE
  transactions on information theory}, vol.~49, no.~2, pp. 371--381, 2003.

\bibitem{koetter2003algebraic}
R.~Koetter and M.~M{\'e}dard, ``An algebraic approach to network coding,''
  \emph{IEEE/ACM Transactions on Networking (TON)}, vol.~11, no.~5, pp.
  782--795, 2003.

\bibitem{ho:medard:koetter:karger:effros:2006}
T.~Ho, M.~M{\'e}dard, R.~Koetter, D.~Karger, M.~Effros, J.~Shi, and B.~Leong,
  ``{A random linear network coding approach to multicast},'' \emph{{IEEE}
  Trans. Inf. Theory}, vol.~52, no.~10, pp. 4413--4430, 2006.

\bibitem{nistor:lucani:vinhoza:costa:barros:2011}
M.~Nistor, D.~E. Lucani, T.~T.~V. Vinhoza, R.~A. Costa, and J.~Barros, ``On the
  delay distribution of random linear network coding,'' \emph{{IEEE} J. Sel.
  Areas Commun.}, vol.~29, no.~5, pp. 1084--1093, May 2011.

\bibitem{lucani:tdd_field:2009}
D.~E. Lucani, M.~M{\'e}dard, and M.~Stojanovic, ``Random linear network coding
  for time-divison duplexing: Field size consideration,'' in \emph{Proc. IEEE
  Global Communications Conference (GLOBECOM)}, 2009.

\bibitem{heide_systematic_RLNC}
J.~Heide, M.~V. Pedersen, F.~H.~P. Fitzek, and T.~Larsen, ``Network coding for
  mobile devices - systematic binary random rateless codes,'' in \emph{Proc.
  IEEE Int. Conf. Communications (ICC) workshop}, 2009.

\bibitem{kwan2011generation}
H.~Y. Kwan, K.~W. Shum, and C.~W. Sung, ``Generation of innovative and sparse
  encoding vectors for broadcast systems with feedback,'' in \emph{Proc. IEEE
  Int. Symp. Information Theory (ISIT)}, 2011, pp. 1161--1165.

\bibitem{yu2014deterministic}
M.~Yu, P.~Sadeghi, and N.~Aboutorab, ``On deterministic linear network coded
  broadcast and its relation to matroid theory,'' in \emph{IEEE Information
  Theory Workshop (ITW)}, 2014, pp. 536--540.

\bibitem{keller2008online}
L.~Keller, E.~Drinea, and C.~Fragouli, ``Online broadcasting with network
  coding,'' in \emph{Proc. of NetCod}, 2008.

\bibitem{yu:sprintson:sadeghi:netcod2015}
M.~Yu, A.~Sprintson, and P.~Sadeghi, ``On minimizing the average packet
  decoding delay in wireless network coded broadcast,'' in \emph{Proc. IEEE
  Int. Symp. Network Coding (NetCod)}, 2015, pp. 1--5.

\bibitem{Rozner_Heuristic_clique}
E.~Rozner, A.~P. Iyer, Y.~Mehta, L.~Qiu, and M.~Jafry, ``{ER}: Efficient
  retransmission scheme for wireless {LAN}s,'' in \emph{Proc. ACM CoNEXT},
  2007.

\bibitem{sorour:valaee:2010}
S.~Sorour and S.~Valaee, ``On minimizing broadcast completion delay for
  instantly decodable network coding,'' in \emph{Proc. IEEE Int. Conf.
  Communications (ICC)}, May 2010, pp. 1--5.

\bibitem{sadeghi:shams:traskov:2010}
P.~Sadeghi, R.~Shams, and D.~Traskov, ``An optimal adaptive network coding
  scheme for minimizing decoding delay in broadcast erasure channels,''
  \emph{EURASIP J. on Wireless Commun. and Netw.}, pp. 1--14, Jan. 2010.

\bibitem{sorour2015completion}
S.~Sorour and S.~Valaee, ``Completion delay minimization for instantly
  decodable network codes,'' \emph{IEEE/ACM Trans. Networking}, vol.~23, no.~5,
  pp. 1553--1567, 2015.

\bibitem{yu2015SIDNC}
M.~Yu, P.~Sadeghi, and N.~Aboutorab, ``Performance characterization and
  transmission schemes for instantly decodable network coding in wireless
  broadcast,'' \emph{European Journal on Advances in Signal Processing
  (EURASIP}, vol. 2015, no.~1, p.~94, 2015.

\bibitem{costa:munaretto:widmer:baros:2008}
R.~Costa, D.~Munaretto, J.~Widmer, and J.~Barros, ``Informed network coding for
  minimum decoding delay,'' in \emph{Proc. IEEE Int. Conf. on Mobile Ad Hoc and
  Sensor System, (MASS)}, 2008, pp. 80--91.

\bibitem{li:idnc_video:2011}
X.~Li, C.-C. Wang, and X.~Lin, ``On the capacity of immediately-decodable
  coding schemes for wireless stored-video broadcast with hard deadline
  constraints,'' \emph{IEEE J. Sel. Areas Commun.}, vol.~29, no.~5, pp.
  1094--1105, 2011.

\bibitem{neda:parastoo:o2idnc}
N.~Aboutorab, S.~Sorour, and P.~Sadeghi, ``O2-{GIDNC}: Beyond instantly
  decodable network coding,'' in \emph{Proc. IEEE Int. Symp. Network Coding
  (NetCod)}, 2013, pp. 1--6.

\bibitem{neda:parastoo:sameh:balance2013}
N.~Aboutorab, P.~Sadeghi, and S.~Sorour, ``Enabling a tradeoff between
  completion time and decoding delay in instantly decodable network coded
  systems,'' \emph{{IEEE} Trans. Commun.}, vol.~62, no.~4, pp. 1296--1309, Mar.
  2014.

\bibitem{sorour2015graph}
S.~Sorour, N.~Aboutoraby, T.~Y. Al-Naffouri, and M.-S. Alouini, ``A graph model
  for opportunistic network coding,'' in \emph{Proc. IEEE Symp. Network Coding
  (NetCod)}, 2015, pp. 26--30.

\bibitem{yu2013rapprochement}
M.~Yu, N.~Aboutorab, and P.~Sadeghi, ``Rapprochement between instantly
  decodable and random linear network coding,'' in \emph{Proc. IEEE Symp.
  Information Theory(ISIT)}, 2013, pp. 3090--3094.

\bibitem{yu:parastoo:neda:2014}
M.~Yu, P.~Sadeghi, and N.~Aboutorab, ``From instantly decodable to random
  linear network coding,'' \emph{{IEEE} Trans. Commun.}, vol.~62, no.~11, pp.
  3943--3955, Oct. 2014.

\bibitem{yu2016feedback}
M.~Yu, P.~Sadeghi, and A.~Sprintson, ``Feedback-assisted random linear network
  coding in wireless broadcast,'' in \emph{Proc. IEEE Global Communications
  Conference (GLOBECOM) Workshop}, 2016, pp. 1--6.

\bibitem{sprintson:min:2007}
S.~El~Rouayheb, M.~A.~R. Chaudhry, and A.~Sprintson, ``On the minimum number of
  transmissions in single-hop wireless coding networks,'' in \emph{Proc. IEEE
  Information Theory Workshop (ITW)}, Lake Tahoe, NV, 2007.

\bibitem{sprintson:algorithm:2008}
M.~A.~R. Chaudhry and A.~Sprintson, ``Efficient algorithms for index coding,''
  in \emph{Proc. IEEE INFOCOM Workshops}, 2008, pp. 1--4.

\bibitem{sprintson:ic_nc_matroid:2010}
S.~El~Rouayheb, A.~Sprintson, and C.~Georghiades, ``On the index coding problem
  and its relation to network coding and matroid theory,'' \emph{{IEEE} Trans.
  Inf. Theory}, vol.~56, no.~7, pp. 3187--3195, 2010.

\bibitem{Yossef:index:2011}
Z.~Bar-Yossef, Y.~Birk, T.~Jayram, and T.~Kol, ``Index coding with side
  information,'' \emph{{IEEE} Trans. Inf. Theory}, vol.~57, no.~3, pp.
  1479--1494, Mar. 2011.

\bibitem{fragouli:pliable:ic:2013}
S.~Brahma and C.~Fragouli, ``Pliable index coding: The multiple requests
  case,'' in \emph{Proc. IEEE Int. Symp. Information Theory (ISIT)}, 2013, pp.
  1142--1146.

\bibitem{sundararajan:sadeghi:medard:2009}
J.~K. Sundararajan, P.~Sadeghi, and M.~M{\'e}dard, ``{A feedback-based adaptive
  broadcast coding scheme for reducing in-order delivery delay},'' in
  \emph{Proc. 5th Workshop on Network Coding, Theory, and Applications
  (NetCod)}, 2009.

\bibitem{li2011capacity}
X.~Li, C.-C. Wang, and X.~Lin, ``On the capacity of immediately-decodable
  coding schemes for wireless stored-video broadcast with hard deadline
  constraints,'' \emph{IEEE Journ. Sel. Areas in Comm.}, vol.~29, no.~5, pp.
  1094--1105, 2011.

\bibitem{le2013instantly}
A.~Le, A.~S. Tehrani, A.~G. Dimakis, and A.~Markopoulou, ``Instantly decodable
  network codes for real-time applications,'' in \emph{Proc. IEEE Symp. Network
  Coding (NetCod)}, 2013, pp. 1--6.

\bibitem{lucani:tdd_rlnc:2009}
D.~E. Lucani, M.~Muriel, and M.~Stojanovic, ``Broadcasting in time-division
  duplexing: A random linear network coding approach,'' in \emph{Proc. 5th
  Workshop on Network Coding, Theory, and Applications (NetCod)}, 2009.

\bibitem{lucani:medard:stojanovic:2010}
D.~Lucani, M.~M{\'e}dard, and M.~Stojanovic, ``Systematic network coding for
  time-division duplexing,'' in \emph{proc. IEEE International Symposium on
  Information Theory Proceedings (ISIT)}, Jun. 2010, pp. 2403--2407.

\bibitem{lucani:medard:stojanovic:2012}
------, ``On coding for delay -- {N}etwork coding for time division
  duplexing,'' \emph{{IEEE} Trans. Inf. Theory}, vol.~58, no.~4, pp.
  2330--2348, Apr. 2012.

\bibitem{el2010coding}
S.~El~Rouayheb, A.~Sprintson, and P.~Sadeghi, ``On coding for cooperative data
  exchange,'' in \emph{IEEE Information Theory Workshop (ITW)}, 2010, pp. 1--5.

\bibitem{sprintson2010randomized}
A.~Sprintson, P.~Sadeghi, G.~Booker, and S.~El~Rouayheb, ``A randomized
  algorithm and performance bounds for coded cooperative data exchange,'' in
  \emph{IEEE Int. Symp. Information Theory (ISIT)}, 2010, pp. 1888--1892.

\end{thebibliography}
\end{document}